%% file: main.tex
\newif\ifsocg
\socgtrue

\ifsocg
    \documentclass[a4paper, UKenglish, cleveref, autoref,thm-restate]{lipics-v2021}
    \hideLIPIcs
    \nolinenumbers 
\else 
    \documentclass{article}
\fi

\usepackage{enumerate}
\input{pre}

\title{Convexity Helps Iterated Search in 3D}

\author{Peyman Afshani}{Department of Computer Science, Aarhus University}{peyman@cs.au.dk}{}{[Supported by DFF (Danmarks Frie Forskningsfond) of Danish Council for Independent Research under grant ID 10.46540/3103-00334B]}
\author{Yakov Nekrich}{Department of Computer Science, Michigan Technological University}{yakov@mtu.edu}{}{}
\author{Frank Staals}{Department of Information and Computing Sciences, Utrecht University, The Netherlands}{f.staals@uu.nl}{0009-0004-8522-1351}{}

\Copyright{Peyman Afshani, Yakov Nekrich, and Frank Staals} 
\authorrunning{P. Afshani, Y. Nekrich, and F. Staals}



\EventEditors{Oswin Aichholzer and Haitao Wang}
\EventNoEds{2}
\EventLongTitle{41st International Symposium on Computational Geometry
(SoCG 2025)}
\EventShortTitle{SoCG 2025}
\EventAcronym{SoCG}
\EventYear{2025}
\EventDate{June 23--27, 2025}
\EventLocation{Kanazawa, Japan}
\EventLogo{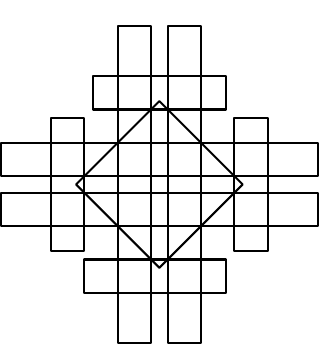}
\SeriesVolume{332}
\ArticleNo{XX}     

\begin{document}

\maketitle

\begin{abstract}
    Inspired by the classical fractional cascading technique~\cite{Chazelle.Guibas.fractional.I,Chazelle.Guibas.fractional.II},
    we introduce new techniques to speed up the following type of iterated search in 3D:
    The input is a graph $\G$ with bounded degree together with a set $H_v$ of 3D hyperplanes associated with every
    vertex of $v$ of $\G$.
    The goal is to store the input such that given a query point $q\in \R^3$ and a connected subgraph
    $\bH\subset \G$, we can decide if $q$ is below or above the lower envelope of $H_v$ for every $v\in \bH$.
    We show that using linear space, it is possible to answer queries in  roughly $O(\log n + |\bH|\sqrt{\log n})$ time which improves
    trivial bound of $O(|\bH|\log n)$ obtained by using planar point location data structures.
    Our data structure can in fact answer more general queries (it combines with shallow
    cuttings) and it even works when 
    $\bH$ is given one vertex at a time.
    We show that this has a number of new applications and in particular,  we give improved solutions to a 
    set of natural 
    data structure problems that up to our knowledge had not seen any improvements.

    We believe this is a very surprising result because obtaining similar results for the planar point
    location problem was known to be impossible~\cite{chazelle2004lower}.

\end{abstract}

\input{intro}

\input{prelim}

\input{tools}

\input{apps}

\bibliography{ref}
\bibliographystyle{plainurl}

\end{document}

%% file: pre.tex
\newtheorem{problem}{Problem}
\usepackage{yfonts}

\ifsocg

\else

    \usepackage{fullpage}
    \usepackage[noend]{algpseudocode}
    \usepackage{xcolor}

    \usepackage{hyperref}

    \usepackage{graphicx}

    \usepackage{amsmath}
    \usepackage{amssymb}
    \usepackage{amsthm}

    \newtheorem{theorem}{Theorem}[section]
    \newtheorem{lemma}[theorem]{Lemma}

\fi

\usepackage{xspace}

\newcommand{\LL}{\log\log}

\renewcommand{\H}{\ensuremath{\mathcal{H}}\xspace}

\newcommand{\A}{\mathcal{A}}
\renewcommand{\P}{\mathcal{P}}
\newcommand{\C}{\mathscr{C}}
\renewcommand{\S}{\mathcal{S}}
\newcommand{\R}{\mathbb{R}}
\newcommand{\E}{\mathbb{E}}

\newtheorem{ob}{Observation}
\newtheorem{prob}{Problem}

\newcommand{\G}{\ensuremath{\mathbf{G}}\xspace}
\newcommand{\bH}{\ensuremath{\mathbf{H}}\xspace}

\newcommand{\N}{\ensuremath{\mathcal{N}}}

\newcommand{\refob}[1]{Observation~\ref{ob:#1}}
\newcommand{\reflem}[1]{Lemma~\ref{lem:#1}}
\newcommand{\refthm}[1]{Theorem~\ref{thm:#1}}

\newcommand{\mypara}[1]{\subparagraph{#1}}

%% file: intro.tex
\section{Introduction}

The idea of ``iterated searching'' was introduced with  the classical fractional cascading technique~\cite{Chazelle.Guibas.fractional.I,Chazelle.Guibas.fractional.II} and it can be described as follows:
consider an abstract data structure problem where given an input $I$, we need to preprocess it 
to answer a given $q$ more efficiently. 
Now imagine that instead of having one input set, $I$, we have multiple independent 
input sets $I_1, I_2, \cdots$ and at the query
time, we would like to answer the same query $q$ on \textit{some} of the data sets. 
This is typically modeled as follows: the input contains a catalog graph $\G$ and each
vertex $v \in \G$ is associated with a different input set $I_v$. 
The goal is to preprocess the input, that is $\G$ and the sets $I_v$, so that given a query $q$ 
and a connected subgraph $\bH \subset \G$ we can find the answer to querying $q$ on 
every data set $I_v$ for each $v \in \bH$.
The trivial solution is to 
have a separate data structure for every $I_v$ and to query them
individually. 
Thus, the main research interest is to improve upon this trivial solution. 

\subsection{Fractional Cascading}

Chazelle and Guibas  improved upon the aforementioned ``trivial solution''
for the predecessor search problem.
To be more precise, for each $v \in\G$, $I_v$
 is an ordered list of values and the query $q$ is also a value and the 
goal is to find the predecessor of $q$ in $I_v$ for every $v\in \bH$. 
Chazelle and Guisbas showed how to build  a data structure that uses linear space to  answer
such queries in $O(\log n + |\bH|)$ time~\cite{Chazelle.Guibas.fractional.I,Chazelle.Guibas.fractional.II},
essentially, reducing the cost of a predecessor search to $O(1)$, after an initial investment of $O(\log n)$ time.
This is clearly optimal and since its introduction has found plenty of applications.

The origin of this technique can be traced back to 1982 and to a paper by 
 Vaishnai and Wood~\cite{Vaishnavi} and  also
the notion of ``downpointers'' from  a 1985 paper by Willard~\cite{Willard85} 
but the seminal work in this area was the aforementioned work by Chazelle and Guibas
\cite{Chazelle.Guibas.fractional.I,Chazelle.Guibas.fractional.II}.

\subsection{2D Fractional Cascading}
Due to its numerous applications, 
Chazelle and Liu \cite{chazelle2004lower} studied ``2D fractional cascading'' which is the iterated
2D point location problem. 
In particular, 
each vertex $v \in \G$ is associated with a planar subdivision, the query
$q$ is a point in $\R^2$,  and the goal is to output for every $v \in \bH$
the cell of the subdivion of $v$ that contains $q$. 
Chazelle and Liu showed that in this case all hope is lost because
there is an almost quadratic space lower bound for data structures
with $O(\log n+ |\bH|)$ query time. 
Actually, their techniques prove  an
$\Omega(n^{2-\varepsilon})$ space lower bound for data structures with 
$O(\log^{O(1)}n) + o(|\bH|\log n)$ query time. 
As $O(|\bH|\log n)$ is the trivial upper bound, the lower bound ``dashes all such hopes''
(as Chazelle and Liu put it)
for obtaining a general technique in dimensions two and beyond.
However, the lower bound only applies to non-orthogonal subdivisions and only for data structures for
sub-quadratic  space consumption. 
Relatively recently, the barrier has been overcome via exploring those two directions.
Afshani and Cheng~\cite{AC.2dfc} studied various cases involving orthogonal subdivisions and obtained almost tight
results, e.g., 
the problem can be solved 
using linear space and with roughly $O(\log n + |\bH|\sqrt{\log n})$ query time for axis-aligned subdivisions.
Chan and Zheng~\cite{hopcroft.timothy22}  studied the non-orthogonal version but they allowed 
quadratic space usage and  they showed that queries can be answered
in $O(\log n + |\bH|)$ time. 
Finally, we need to mention that a more general form of iterated search which allows $\bH$ to be
any subgraph of $\G$ has also been studied for 2D queries 
but the bounds are weaker~\cite{concur.soda14}.

\subsection{A Brief Statement of Our Results}

We introduce a general technique to solve the iterated search problem for a certain 3D query problem.
Then, we show that our technique has a number of non-trivial applications.

\subsubsection{A General Framework for Iterated Searching for Lower Envelopes}

The exact statement of our technique requires familiarity with concepts such as shallow cuttings and therefore
we postpone it to after the ``Preliminaries'' section.
Here, we only state it informally:
we study the \textit{lower envelope iterated search (LEIS) problem}, where we are given a catalog
graph $\G$ and each vertex $v\in\G$ is associated with a catalog which
is a set $H_v$ of hyperplanes.
The goal is to preprocess the input into a data structure
such that given a query point $q$ and a connected subgraph $\bH\subset \G$,
the data structure should report whether $q$ is below or above the lower
envelope of $H_v$, for all $v \in \bH$.
In the shallow cutting variant of LEIS, the input has a parameter $k$ and the data structure stores a number of lists
of size $O(k)$ explicitly. Then at the query time the
data structure has two options: (i) report that $q$ is above the $k$-level of $H_v$  or (ii)
return a constant number, $c$, of pointers to $c$ stored lists $L_1, \cdots, L_c$  such that the union of $L_1, \cdots, L_c$ contains all the 
hyperplanes that pass below $q$.
In the \textit{on-demand} version of the problem, $\bH$ is given one vertex $v$ at a time such that 
all the given vertices form a connected subgraph of $\G$ but the query
must answer $q$ on $v$ before the next vertex is
given. See Figure~\ref{fig:problem_def} for an illustration
of our problem.

Our main result is that all the above variants can be solved with a linear-space data structure that has
$O(\log n + |\bH|\sqrt{\log n})$ query time.
This improves upon the trivial solution which results in $O(|\bH|\log n)$ query time.

We believe this is a surprising result since a LEIS query (which is very commonly used in combination
with shallow cuttings, see~\reflem{shallow}) is typically answered via a point location query
and as Chazelle and Liu's lower bound~\cite{chazelle2004lower} shows, it is hopeless to have an efficient
data structure for iterated point location queries.

\begin{figure}[tb]
  \centering
  \includegraphics{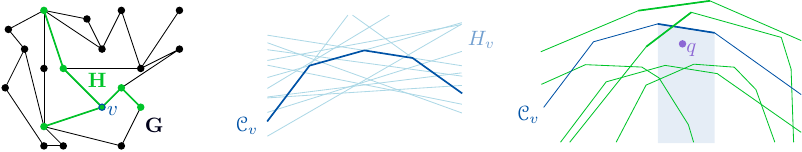}
  \caption{A schematic drawing of the LEIS Problem (with hyperplanes
    in $\R^2$ rather than $\R^3$). Each vertex $v$ of $\G$ is associated
    with a set of hyperplanes $H_v$. For each vertex $v$ in the query
    subgraph $\bH \subset \G$ we wish to find the triangle in the
    $k$-shallow cutting $\C_v$ above the query point $q$.
  }
  \label{fig:problem_def}
\end{figure}

\subsubsection{Some Applications of Our Technique.}\label{sec:weight}

Chazelle and Liu~\cite{chazelle2004lower} cite numerous potential applications of iterated 2D point location,
if such a thing were to exist, before presenting their lower bound.
The lower bound applies to some of those applications but not to all of them.
In particular, it turns out that for some geometric problems, we are able to give improvements.
One main problem is answering \textit{halfspace max queries} which currently has only 
a basic ``folklore'' solution in 3D, despite being used in ``blackbox'' reductions in at least
two general frameworks (more on this below).

\begin{problem}[3D halfspace max queries]
    Let $P$ be an input set of $n$ points in $\R^3$ each associated with a real-valued
    weight. Store them in a data structure such that given a query halfspace $q$ one can
    find the point with the largest weight inside $q$.
\end{problem}

Currently, we are only aware of a ``folklore'' solution that uses $O(n \log n)$ space and has the query
time of $O(\log^2 n)$~\cite{AP.range.sampling}.
Using our framework, we get the following improvements.
\begin{itemize}
    \item 3D halfspace max queries can be answered in $O(\log^{3/2}n)$ time with $O(n\log n)$ space.
    \item The current best solution for \textit{3D halfspace range sampling} uses a 
        ``blackbox'' data structure for weighted range sampling queries~\cite{AP.range.sampling}.
        As a result, we improve the query time of the data structure in~\cite{AP.range.sampling} from $O(\log^2+k)$ (for sampling $k$ points from a given query halfspace)
        to $O(\log^{3/2}n+k)$ by simply plugging in our solution for 3D halfspace max queries.
    \item Range sampling queries can be used as a blackbox to obtain solutions for 
        other approximate aggregate queries~\cite{afshani2023range} and thus our improve carries over to such queries as well.
    \item We get improved results for \emph{colored halfspace
        reporting}, in which we store a set of $n$ colored points so that given a
     query halfspace we can efficiently report all $k$ distinct colors of the
     points that appear in the halfspace. We can achieve $O(\log n\LL n +
     m\LL n)$ query time using linear space, where $m$ is the number
     of colors, or $O(\log n + k\log(m/k)\sqrt{\log n})$ time using
     $O(n\log n)$ space. This improves the query time of current solutions by a
     $O(\log n)$, respectively $O(\sqrt{\log n})$ factor.
   \item Some weighted range reporting problems can also be solved more efficiently:
        (\textit{weighted halfspace range reporting}) store a given set of weighted points such that given a query halfspace $q$ and an interval $[w_1, w_2]$,
        report all the points inside $q$ whose weight is between $w_1$ and $w_2$.
        (\textit{weighted halfspace stabbing}) This is similar to the previous problem with the role of points and halfspaces swapped.
    \item The general framework of~\cite{RahulT2022topk} combines weighted reporting and max reporting to give a general
        solution for ``top-$k$ reporting'' (report the $k$ heaviest points).
        Using our results, we can obtain a data structure for 3D halfspaces with $O(\log^{3/2}n+k)$ query time.
        This was the result ``missing'' in~\cite{RahulT2022topk} where they only list new results for $d=2$ and $d\ge 4$.

    \item Range reporting can also be used to obtain ``approximate counting'' results and a number of general reductions
        are available~\cite{rahul2017approximate,aronov2008approximating,har2011relative,afshani2007approximate} and so
        we can obtain improved results for ``approximate counting'' version of the above weighted reporting problems.

    \item The \textit{offline} version of halfspace range reporting can also be improved: given a set $P$ of $n$ points, where
    a point $p_i$ is given with an insertion time $s_i$ and a deletion time $d_i$, store them in a data structure such that
    given a query halfspace $h$ and a timestamp $t$, we can report all the points in $h$ that exist in time $t$.
    \item Via standard lifting maps, problems involving circles in 2D can be reduced to those involving halfspaces in 3D.
\end{itemize}


%% file: prelim.tex
\section{Preliminaries}

We adopt the following notations and definitions.
In $\R^3$, a \textit{triangle} is the convex hull of three points whereas 
a \textit{simplex} is the convex hull of four points.
The $Z$-coordinate is assumed to be the vertical direction and
a \textit{vertical prism} $\tau^{\downarrow}$ is defined by a triangle $\tau$ and it consists of all the points that are on $\tau$ or directly 
below $\tau$ (wrt to the $Z$-coordinate).
For a point $q \in \R^3$,  $q^{\downarrow}$ represents a downward ray starting from $q$. 
Given a set $A$ of geometric objects and another geometric object $r$, the 
\textit{conflict list} of $r$ wrt $A$, denoted by $\Diamond(A,r)$, is the subset of $A$ that intersects $R$.
For example, if $A$ is a set of polytopes and $r$ is a  point, then $\Diamond(A,r)$ is the set of
of polytopes that contain $r$ or
if $A$ is a set of hyperplanes and $r$ is a point, then $\Diamond(A,r^{\downarrow})$ is the
subset of hyperplanes that pass through or below $r$;
in this latter case, 
$|\Diamond(A,r^{\downarrow})|$ is denoted as the \textit{level} of $r$ (wrt $A$).
For a given parameter $k$, $1 \le k \le n$, the $(\le k)$-level of $A$ is the set of all the points
in $\R^3$ with level at most $k$ and the $k$-level of $A$ is defined as the boundary of the $(\le k)$-level; 
in particular, the $0$-level of $A$ is the lower envelope of $A$. 

\begin{ob}\label{ob:below}\cite{c00}
    Let $\tau$ be a triangle and let $\Delta=\tau^{\downarrow}$ be the vertical prism defined by it.
  Let $v_1, v_2$, and $v_3$ be the vertices of $\tau$.
  Any hyperplanes that intersects $\Delta$, passes below at least one of the vertices of $\tau$.
  In other words, for any set $H$ of hyperplanes we have
  \[
    \Diamond(H,\Delta) \subset \Diamond(H,v_1^{\downarrow}) \cup\Diamond(H,v_2^{\downarrow}) \cup \Diamond(H,v_3^{\downarrow}).
\]
\end{ob}

We will extensively use shallow cuttings in our techniques. By now, this is a standard tool that is used
very often in various 3D range searching problems.

\begin{lemma}\label{lem:shallow}\cite{chan.shallow.cutting,Ramos.SOCG99,Matousek.reporting.points}
    There exists a fixed constant $\alpha > 1$ such that 
    for any given set $H$ of $n$ hyperplanes in 3D and a given parameter $k$, $1 \le k \le n$, the following holds:
    There exists a set of $O(n/k)$ hyperplanes such that their lower envelope
    lies above the $k$-level of $H$ but below the $(\alpha k)$-level of $H$.
    The lower envelope can be triangulated and for every resulting triangle $\tau$,
    one can also build the conflict list $\Diamond(H,\tau^{\downarrow})$ in $O(n \log n)$ time in total
    (over all triangles $\tau$).
    The shallow cutting $\C$ is taken to be the set of all the triangles and each triangle stores
    a pointer to its explicitly stored conflict list. 

    Finally, one can store the projection of the shallow cutting in a 2D point
    location data structure that uses $O(n/k)$ space such that given a query
    point $q \in \R^3$, in $O(\log(n/k))$ time one can decide whether
    $q$ lies below any triangle in $\C$ and if it does, then one can find the
    triangle directly above $q$, denoted by $\C(q)$.
    The subset of hyperplanes passing below $q$ is a subset of the conflict list of
    $\C(q)$, i.e., 
    $\Diamond(H,q^{\downarrow}) \subset \Diamond(H,\C(q)^{\downarrow})$.

\end{lemma}

\mypara{Remark.}
The above is a ``modern'' statement of shallow cuttings in 3D. These were originally discovered by
Matou\v sek~\cite{Matousek.reporting.points}. However, a simple observation by Chan~\cite{c00} (\refob{below}) allows one
to simplify the structure of shallow cuttings into the lower envelope of some hyperplanes.
We will be strongly using convexity in our techniques so this treatment is extremely crucial for our purposes. 

We now define the main iterated search problems that we will be studying. 

\begin{prob}[LEIS3D]

    Let $\G$ be an input graph of constant degree, called the \underline{catalog graph}, 
    where each vertex $v \in G$ is associated with an input set $H_v$ of hyperplanes
    in $\R^3$.
    Let $n$ be the total input size and let $k$, $1 \le k \le n$, be a parameter given with the input.
  Store the input in a structure to answer queries that are given by a point 
  $q \in \R^3$ and a connected subgraph $\bH \subset \G$. 
  For each vertex $v \in \G$, the data structure should store a $k$-shallow cutting of 
  $H_v$. Then, for each $v \in \bH$ at the query time, it must 
  either (i) return a triangle $\tau$  above $q$ together 
  with a list of  $c$ pointers, for a constant $c$, to $c$ conflict lists $L_1,
  \cdots, L_c$ in the $k$-shallow cutting of $H_v$,  such that
  $\Diamond(H_v,\tau) \subset \cup_{i=1}^c L_i$ or (ii)
  indicate that $q$ lies above the $k$-level of $H_v$.

  In the \underline{anchored} version of the problem, for every $v\in \G$ the data
  structure must explicitly store a set $T_v$ of triangles such that 
  the triangle $\tau_v$ returned at query time must be an element of $T_v$.
  In the \underline{on-demand} version of the problem, $q$ is given first but $\bH$ is initially 
  hidden and is revealed
  one vertex at a time at query time. The data structure must answer the
  query on each revealed vertex before the next vertex can be revealed.
  For simplicity, we assume that $\bH$ is a walk, i.e., each revealed vertex is connected to the
  previously revealed vertex.

\end{prob}

Note for $k=0$, we have a special case where during the query it suffices to
simply return a triangle $\tau$ that is below
the lower envelope and above $q$. 
Also the restriction that $\bH$ must be a walk for the
on-demand version is not strictly necessary in RAM models.
However, removing this restriction creates some technical issues in some variants of the
pointer machine model and thus it is added to avoid 
needless complications. 
Furthermore, a walk is sufficient in all the applications that we will consider. 
Our main technical result is the following.

\begin{restatable}{theorem}{thmgraphsc}\label{thm:graphsc}
    The on-demand LEIS3D problem for an input of size $n$ on catalog graph
    $\G$ can be solved with a data structure that uses $O(n)$ space and has the
    following query time: 
    (i) $O(\log n \log\log n + |\bH|\log\log n)$ if $\G$ is a path
    and (ii) 
  $O(\log n + |\bH|\sqrt{\log n })$ when $\G$ is a graph of constant degree.  
\end{restatable}


%% file: tools.tex
\section{Building the Tools}\label{sec:tools}

In this section, we will prove a number of auxiliary results that will later be combined to build
our main results. 

\subsection{Point Location in 3D Convex Overlays}

The following observation is our starting point. 

\begin{restatable}{lemma}{lemoverlay}\label{lem:overlay}
  Let $\P_1, \cdots, \P_m$ be $m$ convex polytopes in $\R^3$ of total complexity $n$.
  Let $\A$ be their overlay. 
  The overlay $\A$ has $O(nm^2)$ complexity.
\end{restatable}

\begin{proof}
  Observe that the boundaries of each convex polytope is composed of vertcies,
  edges, and faces.
  Let $v$ be a vertex in the overlay $\A$.
  We have the following cases:
  \begin{enumerate}
    \item $v$ is a vertex in one of the original convex polytopes. 
      The number of vertcies of this type is clearly $O(n)$.
    \item $v$ is the result of the intersection an edge $e_i \in \P_i$ and $\P_j$ for
      $i\not = j$.
      Observe that each edge $e_i$ intersects each other convex polytope in at most two points.
      This means that each edge $e_i$ can create at most $m$ vertices of this type.
      Thus, the total number of vertices of this type is $O(nm)$.

    \item Finally, $v$ could be the result of the intersection of three faces $f_i \in \P_i$, $f_j \in \P_j$
      and $f_k \in \P_k$ for three distinct convex polytopes. 
      Observe that in this case, the $v$ is the result of the intersection of an edge $e_{ij}$ and the face $f_k$ where
      $e_{ij}$ is the intersection of the faces $f_i$ and $f_j$. 
      If we can prove that the number of such edges $e_{ij}$ is $O(mn)$, then the lemma follows by the same argument as the
      previous two.
      To prove this claim, observe that the edge $e_{ij}$ is also adjacent to two vertices which are the result of the intersection 
      of either $f_i$ and an edge of $\P_j$ or the intersection of an edge of $\P_i$ and $f_j$. In either case, by the previous arguments,
      the number of edges $e_{ij}$ is bounded by $O(mn)$. \qedhere
  \end{enumerate}
\end{proof}

It is critical to observe that an equivalent lemma for 2D subdivision does not
hold.  For example, the overlay of two subdivisions of the plane into $n$
``tall and thin'' and ``wide and narrow'' slabs has $\Omega(n^2)$ size.  Thus,
the above lemma captures a fundamental impact of adding convexity to the
picture. 
The next lemma shows that we can perform point location in 3D via a simple idea
and a small query overhead.

\begin{restatable}[basic point location]{lemma}{basicpl}\label{lem:basicpl}
    Let $S =\{ \P_1, \cdots, \P_m\}$ be a set of $m$ convex polytopes in $\R^3$ of total complexity $n$.
We can build a data structure that uses $O(n m^2)$ space that does the following: 
    The data structure stores a set $T$ of  ``anchors'' (triangles) such that each anchor is either
    inside a polytope or lies completely outside it.
    Then, given  a query point $q$,
    it can output an anchor $\tau \in T$ such that $\Diamond(S,q)$ is equal
    to the subset of polytopes that contain $\tau$. 
    The time to answer the query is $O(\log n \log m)$.
\end{restatable}



\begin{proof}

    We decompose each $\P_i$ into its upper and lower hull, $U_i$ and $L_i$. 
    We can then turn each $U_i$ into a convex surface as follows: consider the projection of $U_i$ onto the
    $XY$-plane which will yield a convex 2D polygon $\P'_i$. Then, for every
    boundary edge $e$ of $\P'_i$, we
    add a hyperplane $h_e$ that goes through $e$. $h_e$ is made ``almost'' vertical (i.e., its angle with the 
    $XY$-plane is chosen close enough to $\frac{\pi}{2}$) such that the adding
    $h_e$ to the upper hull, for
    every boundary edge $e$, creates a convex surface $U'_i$ which can be represented as a function $\R^2 \rightarrow \R^3$
    (i.e., for every point $(x,y) \in \R^2$, there exists a unique value of $z$ such that $(x,y,z)  \in U'_i$).
    We repeat a similar process with the lower envelopes to create another surface $L'_i$.

    To build the data structure, we overlay all the surfaces in and let $\A$ be the resulting overlay. 
    By \reflem{overlay}, the complexity of the overlay $\A$ is $O(nm^2)$. 

    We decompose the overlay $\A$ into $O(m)$
    ``surfaces'', i.e., $Z$-monotone two-dimensional manifolds, as follows:
    Fix a value $i$. 
    Consider the set of points $q \in \R^3$ that have level $i$, i.e., $R_i= \left\{ q |\; |\Diamond(S,q^{\downarrow})| = i \right\}$.
    Define a surface $\S_i$ as the upper boundary of the closure of $R_i$. 
    Observe that $\S_i$ is a subset of the overlay $\A$ and it will also be a $Z$-monotone surface. 
    Recall that each surface $U'_j$ or $L'_j$ can be
    thought of as a function $\R^2 \rightarrow \R$, i.e., a function that is defined everyone on $\R^2$.
    It thus follows that the level of point
    $p$ can only change when lies on the intersection of two surfaces. 
    Thus, the process of creating layers does not create additional edges or vertices, meaning,
    in total the surfaces $\S_i$ have complexity $O(n m^2)$.

    Consequently, each surface $\S_j$ can be triangulated and also projected to the $XY$-plane
    and stored in a point location data structure.
    The triangles resulting from the triangulation are the triangles stored by the data structure. 

    Now consider queries.
    Given a query $q$, we can find the surface $\S_i$ that lies directly above $q$ using binary search on the at most $2m$ levels.
    Each step of the binary search can be done via a point location query for a total query time of $O(\log m \log n)$.
    This finds the triangle $\tau$ in the overlay $\A$ that lies directly above $q$.
    By construction, $\tau$ is the triangle claimed by the lemma.
\end{proof}

For some applications, the above lemma is sufficient
but for our most general result, we need to remove the factor $\log m$ in the above query time. 

\begin{restatable}[advanced point location]{lemma}{advpl}\label{lem:advpl}
    Let $S =\{ \P_1, \cdots, \P_m\}$ be a set of $m$ convex polytopes in
    $\R^3$ of total complexity $n$,
    where $m = O(2^{\sqrt{\log n}})$.
    We can build a data structure that uses $O(n  2^{O(\sqrt{\log
    n})})$ space that does the following: 
    The data structure stores a set $T$ of  ``anchors'' (triangles) such that each anchor is either
    inside a polytope or lies completely outside it.
    Then, given  a query point $q$,
    it can output an anchor $\tau \in T$ such that $\Diamond(S,q)$ is equal
    to the subset of polytopes that contain $\tau$. 
    The query time is ${O}(\log n)$. 
\end{restatable}

The rest of this subsection is devoted to the proof of \reflem{advpl}.

\subsubsection{Preliminaries}

\mypara{A modified Dobkin-Kirkpatrick hierarchy.}
First, we describe a process of simplifying polytopes in 3D which is a modified version of 
building Dobkin-Kirkpartrick hierarchies~\cite{DK} but in a randomized way. 
We refer to the standard Dobkin-Kirkpartrick simplification as DK-simplification and to our
modified randomized simplification as RDK-simplification. 
Let $\P$ be a convex polytope in 3D. 
An \textit{outer simplification} of $\P$ is obtained by considering $\P$ as the intersection of halfspaces that bound
its faces. 
Then, we select a subset $I$ of independent faces of $\P$ (i.e., no two adjacent faces are selected) and such that each face has 
$O(1)$ complexity and $|I| = \Omega(|\P|)$.
In DK-simplification, $I$ is simply deleted, 
instead in RDK-simplification, we delete 
each element of $I$ with probability 0.5, independently. 
Observe that the deletion of $I'$ is equivalent to attaching $|I'|$ disjoint convex cells of
constant complexity to the faces of $\P$.
To define an \textit{inner simplification} of $\P$ we do a similar process but
this time we consider $\P$ as the convex of hull of its vertices and
select an independent set $I$ of the vertices of $\P$ of size $\Omega(|\P|)$ such that each vertex has constant degree.
In the standard DK-simplification, $I$ is simply deleted, 
but in RDK-simplification, we delete each vertex in $I$  with probability 0.5, independently. 
This creates a convex polytope that is contained inside $\P$ (or shares some boundaries).
This process can be thought of as removing disjoint convex cells of
constant complexity from $\P$. See Figure~\ref{fig:simp} for an illustration.
 
\begin{figure}[ht]
    \centering
    \includegraphics{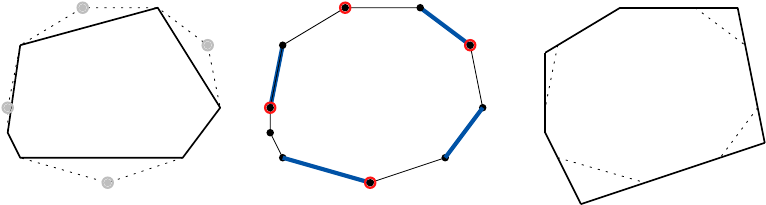}
    \caption{The inner and outer simplifications are shown in 2D, for clarity.
    The central polytope is being simplified here. To the left, we have  its inner simplification after
deleting the vertices marked with red and to the right, we have the outer simplification after deleting
the edges marked with blue. }
    \label{fig:simp}
\end{figure}

We can repeatedly apply DK- or RDK-simplifications until the inner polytope is reduced to a simplex and the outer polytope to a halfspace.
This results in a hierarchy which is denoted by 
$\H(\P)$.  An easy observation that we will be using is the following.

\begin{ob}\label{ob:hierarchy}
    If every simplification reduces the complexity of $\P$ by a constant factor, 
    then, $\H(\P)$ consists of $O(\log n)$ convex polytopes, i.e., it is the overlay of
  $O(\log n)$ convex polytopes which decomposes $\R^3$ into $O(|\P|)$ convex cells of constant complexity. 
\end{ob}

\mypara{A simplification round.} 
Let $T$ be a parameter that will be defined later ($T$ will actually be set to $\sqrt{\log n}$).
Given a convex polytope $\P$, \textit{one round of simplification of $\P$} is defined as follows.
If the complexity of $\P$ is at most $2^T$, then $\P$ is simplified via 
DK-simplification, reducing the complexity of $\P$ by a constant factor at each
simplification to yield a hierarchy of $O(\log |\P|) = O(T)$ simplifications.
Otherwise,  starting from $\P$, we repeatedly apply inner (resp. outer) RDK-simplification $T$ times to 
obtain $T$ progressively smaller  (resp. larger) polytopes. 
The smallest and the largest polytopes that are obtained are called the \textit{boundary} polytopes. 
Based on this, we perform the following steps for every polytope $\P_i$ in the input.

\begin{enumerate}
    \item Initialize a list $X = \{ \P_i \}$ and the set $S = \{\P_i\}$ of simplifications of $\P_i$.
        Each element of $S$ will maintain a \textit{counter} that counts how many rounds it has been simplified for, and
        for $\P_i$ this counter is initialized to zero. 
    \item While $X$ is not empty:
\begin{enumerate}
    \item Remove the first polytope $\P$ from $X$ and let $c$ be the counter of $\P$.
    \item Simplify $\P$ for one round, resulting in $O(T)$ simplifications which
        are added to $S$; all of these simplifications have their counter set to $c+1$.
      \item  If any boundary polytope is produced (i.e., if $\P$ had complexity at least $2^T$),  
          then they are added to $X$. 
\end{enumerate}
\end{enumerate}

We also define a subdivision $\A_j$ is as the overlay of all the
polytopes (over all sets $S$) whose counter is at least $j$.
The following lemma captures some of the important properties of our construction. 

\begin{lemma}\label{lem:basicprop}
    Assuming $T \ge c+\log\log n$ for a large enough constant $c$,
    the following properties hold with probability at least $1-n^2$  (whp). 

    \begin{enumerate}[(I)]
    \item Let $x$ be the maximum value of the counters of the polytopes in $\A_0$. 
      We have   $x = O\left( \frac{\log n}{T} \right)$ (whp). 
    \item Let $M$ be the number of polytopes in $\A_0$. We have $M=O(m
      2^{x} T)$. 
    \item The size of $\A_j$ is $O(n M^2)$, for all $j$.
    \item $\A_j$ decomposes $\R^3$ into convex cells of complexity $O(M)$.
    \item Let $\Delta$ be a cell in $\A_j$ and let $\P$ be a polytope in $\A_{j-1}$. 
        Then, the number of vertices of $\P$ that are contained in 
        $\Delta$ as well as the number of edges of $\P$ that intersect $\Delta$ is bounded by
        $O(2^{O(T)} \log n) = O\left( 2^{O(T)} \right)$. (whp).
    \item The size of $\A_x$ is $O\left( \left( m2^xT 2^T \right)^{O(1)} \right)$.


\end{enumerate}

\end{lemma}

\begin{proof}
  For claim (I),  observe that if a polytope $\P \in X$ has size $2^T$, then it is 
  simplified with a standard DK-simplification and it adds no additional polytopes to $X$.
  Thus, it suffices to consider polytopes that have complexity at least $2^T\ge \log^c n$.
  Now the claim (I) follows from a standard Chernoff bound:
  Since in our RDK-simplification, we will select an independent set $I$ of size $\Omega(|\P|) = \Omega(c\log n)$ and on expectation, we will delete half of elements but since 
  $\E[|I|] \ge c\log n$, it follows that with 
  with high probability at least $|I|/4$ of the elements of $I$ are deleted.
  This reduces the complexity of polytope by a constant fraction.
  Consequently, the number of RDK-simplifications steps an input polytope can go through
  is bounded by $O(\log n)$ and thus the number of rounds is
  $O(\frac{\log  n}{T})$ which implies claim (I).

  Claim (II) follows from the observation that during each round, each
  polytope creates two additional polytopes (whose counters are
  incremented by one).  In other words, the number of polytopes of
  counter $c+1$ is at most double the number of polytopes with counter
  $c$. Combined with (I) this implies that the total number of
  polytopes ever added to $X$ is $O(m 2^x)$. For each such a polytope
  (whose counter is at least $j$) the overlay $\A_j$ includes $T$
  additional polytopes in $S$.

  Claim (III) is a consequence of \reflem{overlay} and (II). 

   To see (IV), consider one polytope $\P$ in $\A_j$.
   Observe that all the inner and outer simplifications of $\P$ are also included in $\A_j$ this is because our simplification process continues until the
   inner simplification of $\P$ is reduced to a simplex and the outer simplification of $\P$ is reduced to a halfspace.
   All of these simplifications have a counter value that is at least $j$ and thus by definition of $\A_j$ they are included in $\A_j$.
   Thus, some DK-hierarchy, $\H(\P)$, of $\P$ is included in $\A_j$.
   By \refob{hierarchy}, $\H(\P)$ decomposes $\R^3$ into cells of constant complexity.
   The number of such hierarchies that are included in $\A_j$ is upper bounded by the number of 
   polytopes, $M$.
   Each cell $\Delta$ of $\A_j$ is thus the intersection of at most $M$ convex cells of constant complexity and thus
   $\Delta$ will have complexity at most $O(M)$.

   Now consider claim (V).
   Observe that if $\P$ is included in $\A_j$, then $\Delta$ cannot contain any vertex of $\P$ or intersect any of its
   edges. 
   Thus, we can assume $\P$ is not included in $\A_j$.
   But this also implies that the counter of $\P$ must have value exactly $j-1$. 
   At some point $\P$ will be considered during our simplification process.
   Assume that $\Delta$ is intersected by $Y$ edges of $\P$ and $Y=\Omega(2^T \log n)$.
   This implies that $\P$ must have undergone our RDK-simplification process.
   Consider one step of our randomized DK-simplification:  regardless of whether we are dealing with inner or outer simplification and regardless of which
   independent set is chosen, each of these $Y$ edges are kept with probability at least $\frac 14$. 
   By a standard Chernoff bound, it follows that with high probability, at least $\frac 18$ fraction of them are kept. 
   $\P$ undergoes at most $T$ such simplification steps until its counter is incremented by one and the resulting boundary polytopes are added to $\A_j$. 
   This implies that with high probability, the resulting boundary polytopes have at least $\frac{Y}{8^T}$ edges intersecting $\Delta$.
   If $Y\ge c 8^T \log n$, this is a contradiction and thus the claim follows. 
   A similar argument applies to the vertices of $\P$ contained in $\Delta$.

   Finally, the claim (VI) follows by the observation that at the last round, all polytopes
   must have complexity at most $2^T$ and there can be at most $m2^xT$ of them by claim (II) and
   now we plug these values in \reflem{overlay} to prove this claim.
\end{proof}

The next lemma enables us to give an upper bound on the number of cells of $\A_{j-1}$ that can intersect the cells of $\A_j$.
This is a very crucial part of our point location data structure but the proof is relatively technical.

\begin{restatable}{lemma}{lemint}\label{lem:int}
  Let $A$ and $B$ be the overlay of at most $L$ convex polytopes each 
  such that they form a decomposition of $\R^3$ into convex cells of complexity at most $\delta$. 
  In addition, assume that for every cell $\Delta$ in $B$, and every polytope $\P$ in $A$,
  the number of vertices of $\P$ that are contained in $\Delta$ and the number of edges of $\P$ that
  intersect $\Delta$ are bounded by $X$.

  Then, every cell $\Delta$ of $B$ is intersected by at most $O(L^3 X
  + L^3 \delta + L^2\delta^2 + L\delta^3 )$ cells of $A$.
\end{restatable}

\begin{proof}

    \begin{figure}[h]
        \centering
        \includegraphics[scale=0.75]{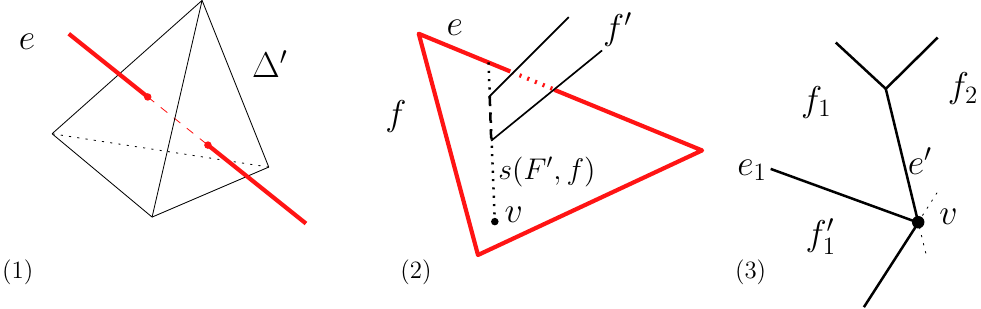}
        \caption{(1) An edge $e$ of $\Delta$ intersects  $\Delta'$.
            (2) A face $f'$ of $\Delta'$ intersects the boundary of $\Delta$ at a face $f$.
            (3) An edge of $\Delta'$ lies completely inside $\Delta$.  $e'$ is the intersection of
            two faces $f_1$ and $f_2$ of potentially two different polytopes. 
    }
        \label{fig:int}
    \end{figure}

  Consider a fixed cell $\Delta$ of $B$ and we count the number of cells $\Delta'$ that can intersect 
  $\Delta$.
  The cells $\Delta'$ and $\Delta$ can intersect at their boundaries
  or one could be contained in the other one. 
  We count the number of such cells $\Delta'$ using the following cases. 
  \begin{enumerate}
    \item An edge $e$ of $\Delta$  intersects $\Delta'$. Observe that every line can intersect each convex polytope
        in $A$ at most twice and thus $e$ can intersect $O(L)$ such cells. As $\Delta$ contains at most $\delta$ edges, 
        the number cells $\Delta'$ can be bounded by $O(L\delta)$.

    \item Now consider the case that a face $f'$  of $\Delta'$ intersects the boundary of $\Delta$, i.e. at a face $f$ of $\Delta$.
      Assume $f'$ lies on a face $F'$ of a polytope $\P'$ in $A$. 
      Let $s(F',f)$ be the line segment that is the intersection of the face $F'$ and the face $f$.
      We now consider two cases:
      (i) If $s(F',f)$ intersects the boundary of the face $f$ at an edge $e$ then it implies that an edge of $\Delta$ (the edge $e$) intersects
      $F'$. The number of such faces $F'$ is at most $O(L\delta)$ by the argument in the previous case and the segment $s(F',f)$ can also intersect at most
      $O(L)$ cells of $A$.
      Thus, the total number of cells $\Delta'$ that can fall in this case is $O(L^2 \delta)$ in total. 
      (ii) Now assume $s(F',f)$ does not intersect the boundary of $\Delta$ which implies it lies completely inside the face $f$ of $\Delta$, meaning,
      it has a vertex $v$ inside $f$.
      But in this case, $v$ is the result of the intersection of an edge of $\P'$ and $f$. 
      By assumptions, the number of such edges is at most $LX$, since there can be at most $X$ edges of every polytope $\P'$ in $A$ intersecting $\Delta$. 
      As before, each of such edge can intersect $O(L)$ cells of $B$ and thus in total the number of cells $\Delta'$ that fall in this case
      is bounded by $O(L^2X)$. 
      Adding up the total from cases (i) and (ii) we get that the number of cells $\Delta'$ in this case is bounded by
      $O(L^2X + L^2 \delta)$.

    \item
        The first two cases cover the situation when the boundaries of $\Delta$ and $\Delta'$ intersect.
        We now consider when one cell contains the other.
        Observe that if $\Delta'$ contains $\Delta$, then $\Delta'$
        can be the only cell that intersects $\Delta$ and thus
        the only case left is when $\Delta'$ is fully contained in $\Delta$.
      In this case, consider one edge of $\Delta'$. 
      This edge will lie on the intersection of two faces $f_1$ and $f_2$ that could potentially belong to different polytopes in $A$. 
      Let $e'$ be the maximal line segment the edge of $\Delta$ in which $f_1$ and $f_2$
      intersect. We have two sub cases:

      \begin{figure}[tb]
        \centering
        \includegraphics{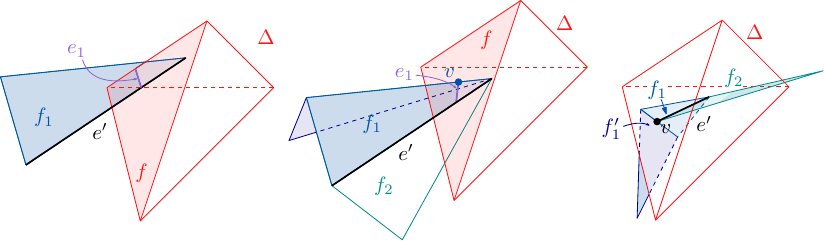}
        \caption{The subcases for case 3: if $e'$ intersects the boundary of $\Delta$ at a face
          $f$ then the situation is essentially the same as in case 2
          (left and middle). If $e'$ has an endpoint $v$ inside
          $\Delta$, it must lie on an edge of a polytope $\P_1$ (right). }
        \label{fig:case3}
      \end{figure}

      \begin{description}
      \item[sub-case: $e'$ intersects the boundary of $\Delta$ at a
        face $f$.] As in case 2, the face $f_1$ must
      intersect $f$, say in an edge $e_1$.

      If $e_1$ intersects an edge of $\Delta$, then it means that
      $f_1$ is of the $O(L\delta)$ faces intersected by
      the edges of $\Delta$ and thus the total number of edges $e_1$
      that belong to this sub-case is $O(L\delta)$. Each such edge
      $e_1$ can intersect $O(L)$ faces as $f_2$, meaning, there are
      $O(L^2\delta)$
      potential edges like $e'$ which in turn implies the total number
      of simplices $\Delta'$ that belong to this sub-case is
      $O(L^2\delta^2)$.

      If $e_1$ does not intersect an edge of $\Delta$, then, as in
      case 2(ii) above, an endpoint $v$ of $e_1$ lies on an edge of a
      polytope $\P'$. By our assumptions there are at most $XL$ such
      edges that intersect $\Delta$. Therefore, there are also at most
      $XL$ edges $e_1$, each of which intersects any other polytope at
      most twice, and thus there are at most $O(L)$ choices for face
      $f_2$. Hence, there are at most $O(XL^2)$ edges $e'$, each of
      which intersects at most $O(L)$ cells $\Delta$. Hence, the
      number of intersections between $\Delta'$ and $\Delta$ of this
      type is $O(XL^3)$.
      
      \item[sub-case: $e'$ does not intersect the boundary of
        $\Delta$.]       This implies that end point $v$ of $e'$ lies  fully inside $\Delta$.
      Now, observe that the faces $f_1$ and $f_2$ are faces of two polytopes $\P_1$ and $\P_2$ and thus the 
      end points of the line segment $e'$ will be the result of the intersection of three faces of $\P_1$ and $\P_2$, meaning, two of the faces must come from
      the same polytope.
      W.l.o.g., assume that $v$ is the result of the intersection of $f_1, f'_1$ and $f_2$ where $f_1$ and $f'_1$ both belong
      to $\P_1$.
      Let $e_1$ be the boundary edge that lies between $f_1$ and $f'_1$. 
      Observe that $e_1$ must either cross the boundary of $\Delta$ or it must have a vertex (of the polytope $\P_1$) inside $\Delta$.
      In either case, we have a bound of $X$ for the number of such edges $e_1$ with respect to polytope $\P_1$ and 
      $O(XL)$ in total. 
      Each such edge $e_1$ can determine $O(L)$ other faces $f_2$ and each such intersection
      can in turn intersect $O(L)$ cells by the same argument.
      Thus the total in this case would be $O(XL^3)$. \qedhere
    \end{description}
  \end{enumerate}
\end{proof}

\subsubsection{Proof of \reflem{advpl} (Advanced Point Location)}

Our data structure is as follows.
Set $T=\sqrt{\log n}$.
Observe that we have $2^T = 2^{\frac{\log n}{T}}$.
We build the previously mentioned overlays $\A_i$ for $i=1, \cdots, O(\log_T n)$.
Then, for each cell $\Delta$ in $\A_j$, we find all the cells $\Delta'$ in $\A_{j-1}$
that intersect $\Delta$ and then store the hyperplanes that define them in what we call a 
\textit{fast data structure}. Given a set of $\ell$ hyperplanes, such
a fast data structure can 
store them using $O(\ell^3)$ space such that point location queries
can be answered in $O(\log\ell)$ time~\cite{Chazelle.cutting}.
This concludes our data structure.
Our query algorithm is relatively simple: if we know the cell $\Delta$ in $\A_j$ that contains the
query point, then by querying the fast data structure built for $\Delta$, we can find the
cell $\Delta'$ in $\A_{j-1}$ that contains the query point and then we can continue until we
finally point locate the query.
We now bound the space and the query time of the data structure.

\mypara{The query time.}
By \reflem{basicprop} (VI), we can locate the query point $q$ in a cell of $\A_x$ 
in $O(\log m + x + \log T + T) = O(\sqrt{\log n})$ time.
Thus, assume, we know $q$ is in a cell $\Delta$ of $\A_j$.
By \reflem{basicprop}, each $\A_j$ will have total complexity $O(n 2^{O(T)})$
and is composed of $2^{O(T)}$ polytopes with cells of complexity $2^{O(T)}$.
By \reflem{int} (using $\delta, L = M, X$ all set to $2^{O(T)}$)
and \reflem{basicprop}(II,IV,V),  $\Delta$ will intersect at most $2^{O(T)}$ cells $\Delta'$, whp.
Note that we can assume this holds in the worst-case because if a cell
intersects more than this amount, then we can
simply rebuild the data structure and
the expected number of such trials is $O(1)$.
Thus, querying the fast data structure will take $O(T)$ time. 
Finally, observe that the depth of the recursion is $x=\frac{\log n}{T}$.
So we get the query time of $O(\log n)$.

\mypara{The space analysis.}
As mentioned each $\A_i$ will have complexity $O(n 2^{O(T)})$.
    For each cell we store a fast data structure that uses $2^{O(T)}$ space, leading to
    $O(n 2^{O(T)})$ total space usage.

\subsection{Solving LEIS3D}

In this section, we prove our main technical tool.

\begin{figure}[tb]
  \centering
  \includegraphics[scale=1.35]{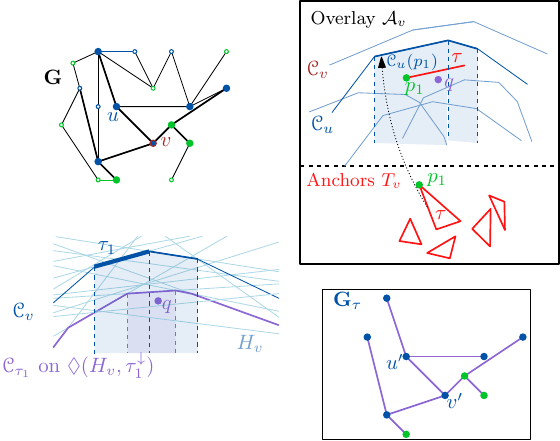}
  \caption{An overview of the data structure. Light nodes shown in
      green, heavy nodes in blue, and $\N_2(v)$ is marked with large disks. 
      Each triangle $\tau$ in the point
    location data structure on $\G_u$ stores a pointer to a copy of
    this neighborhood. From there, we can jump to cells in $\C_v$.}
  \label{fig:overview}
\end{figure}

\thmgraphsc*

\begin{proof}
    Note that when $\G$ is a general graph of maximum degree $d=O(1)$, we can assume that its maximum degree
    is 3 by replacing a vertex of degree $d'\le d$ with a binary tree of depth $\log d'$ which only
    blows up the query time by a constant factor.

    Recall that our data structure should store a $k$-shallow cutting
    of $H_v$ for each vertex $v \in \G$, and let $c$ be a large enough constant.
    When $\G$ is  a path we set $\ell= \log n \cdot \log\log n$ and $x= k (2\ell)^c$ and
    when $\G$ is a general graph, we set $\ell = \sqrt{\log n}$ and $x=k 2^{c \ell}$.
    We say a vertex $u \in \G$ is \textit{light} if $|H_v| \le x$, and
    \textit{heavy} otherwise.
    For a vertex $v$, its \textit{$\ell$-neighborhood}, denoted by $\N_\ell(v)$, is the set of all the vertices of $\G$ that have
    distance at most $\ell$ to $v$.

    \subparagraph{The data structure.}
    For every vertex $v$, we build a $x$-shallow cutting $\C_v$ for the set of hyperplanes $H_v$.
    If $v$ is heavy, $\C_v$ is built via \reflem{shallow} (together with the all the corresponding
    apparatus in the lemma) and thus $\C_v$ is assumed to be convex (i.e., the lower envelope of some hyperplanes).
    However, if $v$ is light, then $\C_v$ is assumed to consist of just one triangle $\tau$,
    the hyperplane $Z=+\infty$ as an infinite triangle and
    the conflict list of $\tau$ in this case is $H_v$.
    In either case, for every triangle $\tau$ in the shallow cutting $\C_v$, we compute and store the conflict
    list, $\Diamond(H_v,\tau^{\downarrow})$, and then we build a $k$-shallow
    cutting, $\C_\tau$, on $\Diamond(H_v,\tau^{\downarrow})$;
    we call $\C_\tau$ \textit{the secondary shallow cutting}.
    This is also built via \reflem{shallow}.
    Note that as $x \geq k$, these secondary shallow cuttings cover the $\leq k$-level of $H_v$.

    \subparagraph{The overlays.}
    Next, we consider each heavy vertex $v$:
    we consider all the vertices $u \in \N_\ell(v)$ and the shallow cuttings $\C_u$ as convex polytopes 
    and create a 3D point location data structure on the overlay all the convex polytopes $\C_u$, for all $u$.
    Let $\A_v$ be this overlay.
    When $\G$ is a path, $\N_\ell(v)$ contains at most $2\ell$ vertices and this overlay is built via \reflem{basicpl}.
    When $\G$ is a graph, it contains at most $2^{\ell+1}$ vertices and we use \reflem{advpl}.
    In either case, the data structure will store a number of triangles as anchors, denoted by $T_v$.
    For each anchor $\tau \in T_v$ we consider its corners (vertices) $p_1, p_2$ and $p_3$. 
    Then, consider every $u \in \N_\ell(v)$ and its shallow cutting $\C_u$;
    store a pointer from $\tau$ to the triangle, $\C_u(p_i)$, in 
    $\C_u$ that lies above $p_i$, for $1 \le i \le 3$.
    Observe that for the anchor $\tau$, we store at most $3\N_{\ell}(v)$ pointers (at most three per $u \in \N_\ell(v)$).
    These pointers are arranged in a particular way:  we make a copy of $\N_\ell(v)$, denoted by $\G_{\tau}$,
    where each vertex $u \in \N_\ell(v)$ has a copy $u' \in \G_{\tau}$ and we store the three pointers
    from $u'$ to $\C_u(p_i)$, for $1 \le i \le 3$.
    
    \subparagraph{The query.}
    Consider a query, consisting of a point $q \in \R^3$ and an initial
    vertex $u$ of the query subgraph $\bH$. 
    \begin{enumerate}
        \item If $u$ is light, $\C_u$ contains one triangle $\tau$ which is the plane $Z=+\infty$ and one cell which is
    the entire $\R^3$.
    We look at the secondary shallow cutting $\C_\tau$ and locate $\C_\tau(q)$ if it exists and
    return conflict list $\Diamond(H_u,\C_\tau(q)^{\downarrow})$ which by
    \reflem{shallow} contains $\Diamond(H_u,q^{\downarrow})$.

        \item 
    This process can be repeated until we encounter the first heavy vertex $u$.
    We query the 3D point location data structure (\reflem{advpl} for general graphs and
    \reflem{basicpl} for when $\G$ is a path) to get an anchor triangle $\tau$ with corners $p_1, p_2$ and $p_3$.
    We use the pointers associated with the copy $u'$ of $u$ in $\G_{\tau}$ to find 
    the triangle $\tau_i = \C_u(p_i)$, $1 \le i \le 3$. 
    The 3D point location data structure guarantees that the corners
    of $\tau$ all lie below $\C_u$, and thus all three of these triangles
    $\tau_i$ exist, if and only if $q$ lies below $\C_u$.

    We query the secondary shallow cuttings $\C_{\tau_i}$, $1 \le i \le 3$, and for each, we
    find $\C_{\tau_i}(q)$, if it exists and then report the pointer to the corresponding secondary 
    conflict list. 

    \item 
    
    Now assume the query is given a sequence of vertices of $\G$, $u_0=u, u_1, \cdots, u_m$, one by one,
    such that they all belong to $\N_\ell(u)$.
    Let $\tau$ be the triangle found in the previous step. 
    The crucial observation is that $\tau$ does not intersect any triangle in the shallow cuttings
    $\C_{w}$, for any $w \in \N_{\ell}(u)$ since $\tau$ was a triangle in the overlay of all the shallow cuttings.
    Thus, in this case, the anchor does not change.
    In this case, for each revealed vertex $u_j$, we simply need to follow the pointer from $u'_{j-1}$ to
    $u'_j$ in $\G_{u,\tau}$ and from $u'_j$  to $\C_{u_j}(p_i)$, for $1 \le i \le 3$.
    The rest of the algorithm is the same as Step 2. 
    \item 
    Eventually, the algorithm will reveal a vertex $u_{m+1}$ that is outside the $\N_\ell(u)$.
    In this case, we simply go to Step 1 or 2, i.e., we assume $u_{m+1}$
    is the initial vertex.
    \end{enumerate}

    \subparagraph{The query time.}
    Step 1 takes $O(\log(\frac{x}{k}))$ time by \reflem{shallow} since each conflict in $\C_v$ has size 
    $O(x)$ and we have a $k$-shallow cutting for them. 
    In Step 2, we use $Q=O(\log n \log\log n)$ time for when $\G$ is a path by \reflem{basicpl} and
    $Q=O(\log n)$ time for a general graph to identify the anchor $\tau$.
    Obtaining the triangles $\C_u(p_i)$ for $1 \le i \le 3$ is a constant time operation and then
    we again query the secondary shallow cuttings which takes $O(\log(\frac{x}{k}))$ time. 
    In Step 3, the query time is only $O(\log(\frac{x}{k}))$ because we do not need to find the anchor
    $\tau$, and the cost is simply the cost of querying the secondary shallow cutting. 
    Finally, observe that whenever we reach Step 2, at least $\ell$ other vertices of $\bH$ must
    be revealed until we exceed $\N_\ell(u)$.
    In other words, Step 2 is only repeated once $\ell$ additional vertices of $\bH$ have been revealed. 
    Summing this up, we get that the query time is asymptotically bounded by 
    \[
        Q+ \log\left(\frac{x}{k}\right)|\bH| + \frac{|\bH|}{\ell} \cdot Q.
    \]
    By plugging in the corresponding parameters, when $\G$ is a path, the query time is
    \[
        O(\log n \log\log n + |\bH|\log\log n)
    \]
    and for a general graph $\G$ it is 
\[
    O(\log n + |\bH|\sqrt{\log n}).
\]

\subparagraph{The space analysis.}
The total size of the shallow cuttings $\C_v$ and the secondary shallow cuttings is $O(n)$.
Let $M$ be the maximum size of $\N_\ell(v)$,
i.e., for a path $\G$ we have $M=2\ell$ and for a general graph of maximum degree three we have
$M=2^\ell$.
Observe that the parameter $x$ is at least $M^c$ in both cases, for some constant $c$ that we can choose. 
The number of heavy vertices is at most $\frac{n}{x}$ and each heavy vertex can appear in the $\ell$-neighborhood of 
at most $M$ other vertices. 
Thus, the total number of overlays that we will create is at most
$\frac{nM}{x}$ and each overlay will contain at most
$M$ convex polytopes. 
By choosing $c$ large enough and by plugging the bounds from \reflem{basicpl} and \reflem{advpl}, we get that the total
space of the 3D point location can be bounded by $O\left( \frac{n}{M} \right)$ which is also an upper bound for the number
of anchors. 
Finally, for every anchor, we store at most $M$ pointers for a total of $O(n)$ space. 
\end{proof}


%% file: apps.tex
\section{Some Applications}
Here, we briefly mention quick applications of \refthm{graphsc}.
Note that by lifting map, similar problems for 2D disks can also be solved.

\begin{theorem}\label{thm:weightrep}
    Let $P$ be a set of $n$ points in 3D, each associated with a real-valued weight.
    We can store $P$ in a data structure that uses $O(n\log n)$ space
    such that given a query halfspace $h$ and two
    values $w_1$ and $w_2$, we can find all the points $P$ that lie inside $h$ and whose
    weight lies between $w_1$ and $w_2$ in $O(\log^{3/2}n + t)$ query time where
    $t$ is the size of the output. 
\end{theorem}

\begin{proof}
    Use duality to get a set of $n$ hyperplanes and let $q$ be point
    dual to the query halfspace $h$.
    Store the hyperplanes  in a balanced binary tree $\G$, ordered by weight. 
    Every node $v$ of $\G$ defines a \textit{canonical set $H_v$} which is the subset of hyperplanes
    stored in the subtree of $v$.
    Store $H_v$ in a data structure of Afshani and Chan~\cite{ac09}.
    This uses $O(n \log n)$ space in total.
    Then, answering the query can be reduced to answering $O(\log n)$ 3D
    halfspace range reporting queries on two root to leaf paths in the tree.
    Simply querying the data structure of~\cite{ac09} gives $O(\log^2 n + t)$
    query time which is already optimal if $t \ge \log^2 n$.
    To improve the query time, it thus suffices to assume $t \le \log^2 n$.
    In which case, we use \refthm{graphsc}  with parameter $k=\log^2 n$ 
    and store each conflict list stored by~\refthm{graphsc} in a data structure of~\cite{ac09}. 
    This combination still uses linear space and when $t \le\log^2n$, it returns
    $O(1)$ conflict lists of size $O(\log^2 n)$ each.
    Querying the data structure of~\cite{ac09} on the conflist lists gives the claimed query time. 
\end{proof}

Note that offline halfspace range reporting queries can also be solved with a very similar technique.
We now consider 3D halfspace max queries.

\begin{theorem}
  \label{thm:weighted_hyperplanes_blelow}
  3D halfspace range max queries can be solved 
    $O(n\log n)$ space and $O(\log^{3/2}n)$ query time.
\end{theorem}

\begin{proof}
    Use duality and let $H=\{h_1,..,h_n\}$ be the dual input hyperplanes ordered by
  increasing weights, and let $w_i$ be the weight of $h_i$. We construct a
  balanced binary tree \G, whose leaves correspond to the hyperplanes
  $h_1,..,h_n$. For each node $v$, let $H_v$ denote the subset of
  hyperplanes stored in the leaves below $v$. We now build the data
  structure of \refthm{graphsc}, with parameter $k=0$.
  Total number of hyperplanes is $O(n\log n)$, the data structure uses $O(n\log n)$ space.

  Consider the query point $q \in \R^3$ dual to the query halfspace, and let $h_i$ be the (unknown)
  heaviest plane that passes below $q$. To find $h_i$, we construct a
  path \bH of length $O(\log n)$ in an on-demand manner.
  The main idea is as follows. At node $v$, with left child $\ell$ and right
  child $r$, we want to test if $q$ lies below the lower envelope (the
  $0$-level) of $H_r$. If so, then $h_i$ must lie in the left subtree
  rooted at $\ell$, hence we continue the search there. If $q$ lies
  above the lower envelope of $H_r$, we continue the search in the
  subtree rooted at $r$. The process finishes after we find the leaf
  containing $h_i$ in $O(\log n)$ rounds. This process generates a
  path \bH of length $O(\log n)$. Hence, the total query time is
  $O(\log^{3/2} n)$ as claimed.
\end{proof}

As discussed, this also improves the query time of the best range sampling data
structures for weighted points.  See~\cite{AP.range.sampling}.
It can also be combined with~\refthm{weightrep} to answer ``top-$k$'' reporting queries
of halfspaces in 3D~\cite{RahulT2022topk} as well as approximate counting queries~\cite{rahul2017approximate}.

\begin{theorem}
  \label{thm:colored-halfspace_reporting}
  Let $P$ be a set of points, each associated with a color from the set $[m]$.
  We can store $P$ in a data structure, s.t., given a query halfspace $h$, we can report the $t$ distinct colors of the points
  in $h$ in (i) $O(\log n \log\log + m \log\log n)$ time
  using $O(n)$ space or
  (ii) $O(\log n + t\log (m/t)\sqrt{\log n})$ time using $O(n \log m)$ space.
\end{theorem}

\begin{proof}
    Once again use duality and let $H$ be the of dual halfspaces. 
    Let $H_i$ be the hyperplanes of color $i$, $1 \le i \le m$. 
    For the first claim (i), simply use a path of length $m$ as graph $\G$
    and invoke \refthm{graphsc} where all query subgraphs $\bH$ are equal to $\G$.

    For the second claim (ii), add a balanced binary tree of height $\log m$ on the path
    and let $\G$ be the resulting graph. 
    For every node $v \in \G$, define $H_v$ to be the union of hyperplanes stored
    in the subtree of $v$. We now build the \refthm{graphsc} data structure
    with $k=0$ the graph \G. 
    The total size of all sets $H_v$, and thus of our data structure is $O(n\log m)$.
    Now, consider the query point $q \in \R^3$ dual to $h$. 
    Once again, we use the on-demand capability and at every node $v \in \G$ we
    decide whether either child of $v$ has any output color, and if so,
    we recurse into that child.
    Let $\bH$ denote the subtree of \G visited by this procedure. 
    By classical results, size of  $\bH$ is
  $O(t\log(n/t))$, when $t$ is the number of leaves visited (and thus
  the number of distinct colors reported). 
  Note that $\bH$ can be assumed to be a walk by simpling fully finishing the recursion on 
  each child of $v$, then backing up to $v$ and then potentially recursing on the other child.
  The total query time is $O(\log n + t\log(m/t)\sqrt{\log n})$.
\end{proof}